\documentclass[11pt]{article}
\usepackage[a4paper]{geometry}
\usepackage{amsfonts, amsmath, amssymb, amsthm, graphicx, caption, authblk, multirow, makecell, framed, float, xcolor, enumitem, tikz, hyperref}
\theoremstyle{definition}

\newtheorem{lemma}{Lemma}

\theoremstyle{remark}

\definecolor{blk}{RGB}{63,63,63}
\newcommand*{\mybox}[1]{%
  \framebox{\raisebox{0cm}[0.5\baselineskip][0.05\baselineskip]{%
    \hbox to 0.10cm {\hss#1\hss}}}\hspace{0.05cm}}
\newcommand{\mystack}[1]{\begin{array}{|r|||}\hline#1\\\hline\end{array}}

\begin{document}
\title{Physical Zero-Knowledge Proof for Ball Sort Puzzle}
\author[1]{Suthee Ruangwises\thanks{\texttt{ruangwises@gmail.com}}}
\affil[1]{Department of Informatics, The University of Electro-Communications, Tokyo, Japan}
\date{}
\maketitle

\begin{abstract}
Ball sort puzzle is a popular logic puzzle consisting of several bins containing balls of multiple colors. Each bin works like a stack; a ball has to follow the last-in first-out order. The player has to sort the balls by color such that each bin contains only balls of a single color. In this paper, we propose a physical zero-knowledge proof protocol for the ball sort puzzle using a deck of playing cards, which enables a prover to physically show that he/she knows a solution with $t$ moves of the ball sort puzzle without revealing it. Our protocol is the first zero-knowledge proof protocol for an interactive puzzle involving moving objects.

\end{abstract}

\section{Introduction}
\textit{Ball sort puzzle} is a logic puzzle which has recently become popular via smartphone apps \cite{google}. The player is given $n$ bins, each fully filled with $h$ balls. All the $hn$ balls are classified into $n$ colors, each with $h$ balls. The player is also given $m$ additional empty bins, each also with capacity $h$. The objective of this puzzle is to ``sort'' the balls by color, i.e. make each bin either empty or full with balls of a single color. Each bin works like a stack; the player can only pick the topmost ball from a bin and place it to the top of another bin that is not full. The additional restriction for each move is that, if the destination bin is not empty, the color of its top ball must be the same as the moved ball. See Fig. \ref{fig0} for an example.

\begin{figure}
\centering
\begin{tikzpicture}
\draw[line width=0.6mm] (0,0) -- (0,1.9);
\draw[line width=0.6mm] (0,0) -- (0.7,0);
\draw[line width=0.6mm] (0.7,0) -- (0.7,1.9);
\draw[line width=0.6mm] (1,0) -- (1,1.9);
\draw[line width=0.6mm] (1,0) -- (1.7,0);
\draw[line width=0.6mm] (1.7,0) -- (1.7,1.9);
\draw[line width=0.6mm] (2,0) -- (2,1.9);
\draw[line width=0.6mm] (2,0) -- (2.7,0);
\draw[line width=0.6mm] (2.7,0) -- (2.7,1.9);
\draw[line width=0.6mm] (3,0) -- (3,1.9);
\draw[line width=0.6mm] (3,0) -- (3.7,0);
\draw[line width=0.6mm] (3.7,0) -- (3.7,1.9);

\node[draw,circle] at (0.35,0.35) {3};
\node[draw,circle] at (0.35,0.95) {1};
\node[draw,circle] at (0.35,1.55) {2};
\node[draw,circle] at (1.35,0.35) {2};
\node[draw,circle] at (1.35,0.95) {2};
\node[draw,circle] at (1.35,1.55) {3};
\node[draw,circle] at (2.35,0.35) {1};
\node[draw,circle] at (2.35,0.95) {1};
\node[draw,circle] at (2.35,1.55) {3};
\node at (4.3,0.95) {\LARGE{$\rightarrow$}};
\end{tikzpicture}
\begin{tikzpicture}
\draw[line width=0.6mm] (0,0) -- (0,1.9);
\draw[line width=0.6mm] (0,0) -- (0.7,0);
\draw[line width=0.6mm] (0.7,0) -- (0.7,1.9);
\draw[line width=0.6mm] (1,0) -- (1,1.9);
\draw[line width=0.6mm] (1,0) -- (1.7,0);
\draw[line width=0.6mm] (1.7,0) -- (1.7,1.9);
\draw[line width=0.6mm] (2,0) -- (2,1.9);
\draw[line width=0.6mm] (2,0) -- (2.7,0);
\draw[line width=0.6mm] (2.7,0) -- (2.7,1.9);
\draw[line width=0.6mm] (3,0) -- (3,1.9);
\draw[line width=0.6mm] (3,0) -- (3.7,0);
\draw[line width=0.6mm] (3.7,0) -- (3.7,1.9);

\node[draw,circle] at (0.35,0.35) {3};
\node[draw,circle] at (0.35,0.95) {1};
\node[draw,circle] at (0.35,1.55) {2};
\node[draw,circle] at (1.35,0.35) {2};
\node[draw,circle] at (1.35,0.95) {2};
\node[draw,circle] at (1.35,1.55) {3};
\node[draw,circle] at (2.35,0.35) {1};
\node[draw,circle] at (2.35,0.95) {1};
\node[draw,circle] at (3.35,0.35) {3};
\end{tikzpicture}\\
\begin{tikzpicture}
\node at (0,0) {\LARGE{$\swarrow$}};
\end{tikzpicture}\\
\begin{tikzpicture}
\draw[line width=0.6mm] (0,0) -- (0,1.9);
\draw[line width=0.6mm] (0,0) -- (0.7,0);
\draw[line width=0.6mm] (0.7,0) -- (0.7,1.9);
\draw[line width=0.6mm] (1,0) -- (1,1.9);
\draw[line width=0.6mm] (1,0) -- (1.7,0);
\draw[line width=0.6mm] (1.7,0) -- (1.7,1.9);
\draw[line width=0.6mm] (2,0) -- (2,1.9);
\draw[line width=0.6mm] (2,0) -- (2.7,0);
\draw[line width=0.6mm] (2.7,0) -- (2.7,1.9);
\draw[line width=0.6mm] (3,0) -- (3,1.9);
\draw[line width=0.6mm] (3,0) -- (3.7,0);
\draw[line width=0.6mm] (3.7,0) -- (3.7,1.9);

\node[draw,circle] at (0.35,0.35) {3};
\node[draw,circle] at (0.35,0.95) {1};
\node[draw,circle] at (0.35,1.55) {2};
\node[draw,circle] at (1.35,0.35) {2};
\node[draw,circle] at (1.35,0.95) {2};
\node[draw,circle] at (2.35,0.35) {1};
\node[draw,circle] at (2.35,0.95) {1};
\node[draw,circle] at (3.35,0.35) {3};
\node[draw,circle] at (3.35,0.95) {3};
\node at (4.3,0.95) {\LARGE{$\rightarrow$}};
\end{tikzpicture}
\begin{tikzpicture}
\draw[line width=0.6mm] (0,0) -- (0,1.9);
\draw[line width=0.6mm] (0,0) -- (0.7,0);
\draw[line width=0.6mm] (0.7,0) -- (0.7,1.9);
\draw[line width=0.6mm] (1,0) -- (1,1.9);
\draw[line width=0.6mm] (1,0) -- (1.7,0);
\draw[line width=0.6mm] (1.7,0) -- (1.7,1.9);
\draw[line width=0.6mm] (2,0) -- (2,1.9);
\draw[line width=0.6mm] (2,0) -- (2.7,0);
\draw[line width=0.6mm] (2.7,0) -- (2.7,1.9);
\draw[line width=0.6mm] (3,0) -- (3,1.9);
\draw[line width=0.6mm] (3,0) -- (3.7,0);
\draw[line width=0.6mm] (3.7,0) -- (3.7,1.9);

\node[draw,circle] at (0.35,0.35) {3};
\node[draw,circle] at (0.35,0.95) {1};
\node[draw,circle] at (1.35,0.35) {2};
\node[draw,circle] at (1.35,0.95) {2};
\node[draw,circle] at (1.35,1.55) {2};
\node[draw,circle] at (2.35,0.35) {1};
\node[draw,circle] at (2.35,0.95) {1};
\node[draw,circle] at (3.35,0.35) {3};
\node[draw,circle] at (3.35,0.95) {3};
\end{tikzpicture}\\
\begin{tikzpicture}
\node at (0,0) {\LARGE{$\swarrow$}};
\end{tikzpicture}\\
\begin{tikzpicture}
\draw[line width=0.6mm] (0,0) -- (0,1.9);
\draw[line width=0.6mm] (0,0) -- (0.7,0);
\draw[line width=0.6mm] (0.7,0) -- (0.7,1.9);
\draw[line width=0.6mm] (1,0) -- (1,1.9);
\draw[line width=0.6mm] (1,0) -- (1.7,0);
\draw[line width=0.6mm] (1.7,0) -- (1.7,1.9);
\draw[line width=0.6mm] (2,0) -- (2,1.9);
\draw[line width=0.6mm] (2,0) -- (2.7,0);
\draw[line width=0.6mm] (2.7,0) -- (2.7,1.9);
\draw[line width=0.6mm] (3,0) -- (3,1.9);
\draw[line width=0.6mm] (3,0) -- (3.7,0);
\draw[line width=0.6mm] (3.7,0) -- (3.7,1.9);

\node[draw,circle] at (0.35,0.35) {3};
\node[draw,circle] at (1.35,0.35) {2};
\node[draw,circle] at (1.35,0.95) {2};
\node[draw,circle] at (1.35,1.55) {2};
\node[draw,circle] at (2.35,0.35) {1};
\node[draw,circle] at (2.35,0.95) {1};
\node[draw,circle] at (2.35,1.55) {1};
\node[draw,circle] at (3.35,0.35) {3};
\node[draw,circle] at (3.35,0.95) {3};
\node at (4.3,0.95) {\LARGE{$\rightarrow$}};
\end{tikzpicture}
\begin{tikzpicture}
\draw[line width=0.6mm] (0,0) -- (0,1.9);
\draw[line width=0.6mm] (0,0) -- (0.7,0);
\draw[line width=0.6mm] (0.7,0) -- (0.7,1.9);
\draw[line width=0.6mm] (1,0) -- (1,1.9);
\draw[line width=0.6mm] (1,0) -- (1.7,0);
\draw[line width=0.6mm] (1.7,0) -- (1.7,1.9);
\draw[line width=0.6mm] (2,0) -- (2,1.9);
\draw[line width=0.6mm] (2,0) -- (2.7,0);
\draw[line width=0.6mm] (2.7,0) -- (2.7,1.9);
\draw[line width=0.6mm] (3,0) -- (3,1.9);
\draw[line width=0.6mm] (3,0) -- (3.7,0);
\draw[line width=0.6mm] (3.7,0) -- (3.7,1.9);

\node[draw,circle] at (1.35,0.35) {2};
\node[draw,circle] at (1.35,0.95) {2};
\node[draw,circle] at (1.35,1.55) {2};
\node[draw,circle] at (2.35,0.35) {1};
\node[draw,circle] at (2.35,0.95) {1};
\node[draw,circle] at (2.35,1.55) {1};
\node[draw,circle] at (3.35,0.35) {3};
\node[draw,circle] at (3.35,0.95) {3};
\node[draw,circle] at (3.35,1.55) {3};
\end{tikzpicture}
\caption{An example of a ball sort puzzle with $h=3$, $n=3$, and $m=1$ with its solution in five moves}
\label{fig0}
\end{figure}

Very recently, Ito et al. \cite{np} proved that it is NP-complete to determine whether a given instance of the ball sort puzzle has a solution with at most $t$ moves for a given integer $t$, or even to determine whether it has a solution at all. They also showed that an instance of the ball sort puzzle is solvable if and only if its corresponding instance of a \textit{water sort puzzle} (a similar puzzle with more restrictive rules) is solvable.

Suppose that Philip, a puzzle expert, constructed a difficult ball sort puzzle and challenged his friend Vera to solve it. After trying for a while, Vera could not solve his puzzle and wondered whether it has a solution or not. Philip has to convince her that his puzzle actually has a solution without revealing it (which would render the challenge pointless). In this situation, Philip needs some kind of \textit{zero-knowledge proof (ZKP)}.

\subsection{Zero-Knowledge Proof}
Introduced by Goldwasser et al. \cite{zkp0}, a ZKP is an interactive protocol between a prover $P$ and a verifier $V$. Both $P$ and $V$ are given a computational problem $x$, but its solution $w$ is known only to $P$. A ZKP enables $P$ to convince $V$ that he/she knows $w$ without leaking any information about $w$. A ZKP with perfect completeness and perfect soundness must satisfy the following three properties.

\begin{enumerate}
	\item \textbf{Perfect Completeness:} If $P$ knows $w$, then $V$ always accepts.
	\item \textbf{Perfect Soundness:} If $P$ does not know $w$, then $V$ always rejects.
	\item \textbf{Zero-knowledge:} $V$ gains no information about $w$, i.e. there is a probabilistic polynomial time algorithm $S$ (called a \textit{simulator}), not knowing $w$ but having access to $V$, such that the outputs of $S$ follow the same probability distribution as the ones of the actual protocol.
\end{enumerate}

As a ZKP exists for every NP problem \cite{zkp}, one can construct a computational ZKP for the ball sort puzzle. However, such a construction requires cryptographic primitives and thus is not practical or intuitive. Instead, many researchers have developed physical ZKPs for logic puzzles using a deck of playing cards. These card-based protocols have benefits that they require only small, portable objects and do not require computers. They also allow external observers to verify that the prover truthfully executes the protocol (which is a challenging task for digital protocols). Moreover, these protocols are suitable for teaching the concept of a ZKP to non-experts.

\subsection{Related Work}
Recently, card-based ZKP protocols for several types of logic puzzles have been developed: Akari \cite{akari}, Bridges \cite{bridges}, Heyawake \cite{nurikabe}, Hitori \cite{nurikabe}, Juosan \cite{takuzu}, Kakuro \cite{akari,kakuro}, KenKen \cite{akari}, Makaro \cite{makaro,makaro2}, Masyu \cite{slitherlink}, Nonogram \cite{nonogram,nonogram2}, Norinori \cite{norinori}, Numberlink \cite{numberlink}, Nurikabe \cite{nurikabe}, Nurimisaki \cite{nurimisaki}, Ripple Effect \cite{ripple}, Shikaku \cite{shikaku}, Slitherlink \cite{slitherlink}, Sudoku \cite{sudoku0,sudoku2,sudoku}, Suguru \cite{suguru}, Takuzu \cite{akari,takuzu}, Usowan \cite{usowan}, and ABC End View \cite{abc}. All of them are pencil puzzles where a solution is a written answer.

\subsection{Our Contribution}
In this paper, we propose a card-based ZKP protocol for the ball sort puzzle with perfect completeness and soundness, which enables the prover $P$ to physically show that he/she knows a solution with $t$ moves of the ball sort puzzle.

By simulating each move by cards, our protocol is the first card-based ZKP protocol for an interactive puzzle (where a solution involves moving objects instead of just a written answer as in pencil puzzles).

\section{Preliminaries}
We will encode an instance of the ball sort puzzle by cards. Our protocol will enable $P$ to simulate each move of a ball such that $V$ can verify that the move is valid but does not know the color of the moved ball, or which bins the ball is moved to and from. In this section, we will introduce subprotocols that are necessary to hide such information.

\subsection{Cards}
Each card used in our protocol has a non-negative integer written on the front side. All cards have indistinguishable back sides denoted by \mybox{?}.

For $1 \leq x \leq q$, define $E_q(x)$ to be a sequence of consecutive $q$ cards, with all of them being \mybox{0}s except the $x$-th card from the left being a \mybox{1}, e.g. $E_4(3)$ is \mbox{\mybox{0}\mybox{0}\mybox{1}\mybox{0}}. This encoding rule was first used by Shinagawa et al. \cite{polygon} in the context of encoding each integer in $\mathbb{Z}/q\mathbb{Z}$ with a regular $q$-gon card.

Furthermore, we define $E_q(0)$ to be a sequence of consecutive $q$ cards, all of them being \mybox{0}s, and $E_q(q+1)$ to be a sequence of consecutive $q$ cards, all of them being \mybox{1}s, e.g. $E_4(0)$ is \mbox{\mybox{0}\mybox{0}\mybox{0}\mybox{0}} and $E_4(5)$ is \mbox{\mybox{1}\mybox{1}\mybox{1}\mybox{1}}.

In some operations, we may stack the cards in $E_q(x)$ into a single stack (with the leftmost card being the topmost card in the stack).

\subsection{Pile-Shifting Shuffle}
Given a $p \times q$ matrix $M$ of cards, a \textit{pile-shifting shuffle} \cite{polygon} shifts the columns of $M$ by a uniformly random cyclic shift unknown to all parties (see Fig. \ref{fig1}). It can be implemented by putting all cards in each column into an envelope, and taking turns to apply \textit{Hindu cuts} (taking several envelopes from the bottom of the pile and putting them on the top) to the pile of envelopes \cite{hindu}.

Note that each card in the matrix can be replaced by a stack of cards, and the protocol still works in the same way as long as every stack in the same row consists of the same number of cards.

\begin{figure}
\centering
\begin{tikzpicture}
\node at (0,0.6) {\mybox{?}};
\node at (0.5,0.6) {\mybox{?}};
\node at (1,0.6) {\mybox{?}};
\node at (1.5,0.6) {\mybox{?}};
\node at (2,0.6) {\mybox{?}};

\node at (0,1.2) {\mybox{?}};
\node at (0.5,1.2) {\mybox{?}};
\node at (1,1.2) {\mybox{?}};
\node at (1.5,1.2) {\mybox{?}};
\node at (2,1.2) {\mybox{?}};

\node at (0,1.8) {\mybox{?}};
\node at (0.5,1.8) {\mybox{?}};
\node at (1,1.8) {\mybox{?}};
\node at (1.5,1.8) {\mybox{?}};
\node at (2,1.8) {\mybox{?}};

\node at (0,2.4) {\mybox{?}};
\node at (0.5,2.4) {\mybox{?}};
\node at (1,2.4) {\mybox{?}};
\node at (1.5,2.4) {\mybox{?}};
\node at (2,2.4) {\mybox{?}};

\node at (-0.4,0.6) {4};
\node at (-0.4,1.2) {3};
\node at (-0.4,1.8) {2};
\node at (-0.4,2.4) {1};

\node at (0,2.9) {1};
\node at (0.5,2.9) {2};
\node at (1,2.9) {3};
\node at (1.5,2.9) {4};
\node at (2,2.9) {5};

\node at (2.9,1.5) {\LARGE{$\Rightarrow$}};
\end{tikzpicture}
\begin{tikzpicture}
\node at (0,0.6) {\mybox{?}};
\node at (0.5,0.6) {\mybox{?}};
\node at (1,0.6) {\mybox{?}};
\node at (1.5,0.6) {\mybox{?}};
\node at (2,0.6) {\mybox{?}};

\node at (0,1.2) {\mybox{?}};
\node at (0.5,1.2) {\mybox{?}};
\node at (1,1.2) {\mybox{?}};
\node at (1.5,1.2) {\mybox{?}};
\node at (2,1.2) {\mybox{?}};

\node at (0,1.8) {\mybox{?}};
\node at (0.5,1.8) {\mybox{?}};
\node at (1,1.8) {\mybox{?}};
\node at (1.5,1.8) {\mybox{?}};
\node at (2,1.8) {\mybox{?}};

\node at (0,2.4) {\mybox{?}};
\node at (0.5,2.4) {\mybox{?}};
\node at (1,2.4) {\mybox{?}};
\node at (1.5,2.4) {\mybox{?}};
\node at (2,2.4) {\mybox{?}};

\node at (-0.4,0.6) {4};
\node at (-0.4,1.2) {3};
\node at (-0.4,1.8) {2};
\node at (-0.4,2.4) {1};

\node at (0,2.9) {4};
\node at (0.5,2.9) {5};
\node at (1,2.9) {1};
\node at (1.5,2.9) {2};
\node at (2,2.9) {3};
\end{tikzpicture}
\caption{An example of a pile-shifting shuffle on a $4 \times 5$ matrix}
\label{fig1}
\end{figure}

\subsection{Pile-Scramble Shuffle}
Given a $p \times q$ matrix $M$ of cards, a \textit{pile-scramble shuffle} \cite{scramble} rearranges the columns of $M$ by a uniformly random permutation unknown to all parties (see Fig. \ref{fig2}). It can be implemented by putting all cards in each column into an envelope, and scrambling all envelopes together completely randomly.

Like the pile-shifting shuffle, each card in the matrix can be replaced by a stack of cards, and the protocol still works in the same way as long as every stack in the same row consists of the same number of cards.

\begin{figure}
\centering
\begin{tikzpicture}
\node at (0,0.6) {\mybox{?}};
\node at (0.5,0.6) {\mybox{?}};
\node at (1,0.6) {\mybox{?}};
\node at (1.5,0.6) {\mybox{?}};
\node at (2,0.6) {\mybox{?}};

\node at (0,1.2) {\mybox{?}};
\node at (0.5,1.2) {\mybox{?}};
\node at (1,1.2) {\mybox{?}};
\node at (1.5,1.2) {\mybox{?}};
\node at (2,1.2) {\mybox{?}};

\node at (0,1.8) {\mybox{?}};
\node at (0.5,1.8) {\mybox{?}};
\node at (1,1.8) {\mybox{?}};
\node at (1.5,1.8) {\mybox{?}};
\node at (2,1.8) {\mybox{?}};

\node at (0,2.4) {\mybox{?}};
\node at (0.5,2.4) {\mybox{?}};
\node at (1,2.4) {\mybox{?}};
\node at (1.5,2.4) {\mybox{?}};
\node at (2,2.4) {\mybox{?}};

\node at (-0.4,0.6) {4};
\node at (-0.4,1.2) {3};
\node at (-0.4,1.8) {2};
\node at (-0.4,2.4) {1};

\node at (0,2.9) {1};
\node at (0.5,2.9) {2};
\node at (1,2.9) {3};
\node at (1.5,2.9) {4};
\node at (2,2.9) {5};

\node at (2.9,1.5) {\LARGE{$\Rightarrow$}};
\end{tikzpicture}
\begin{tikzpicture}
\node at (0,0.6) {\mybox{?}};
\node at (0.5,0.6) {\mybox{?}};
\node at (1,0.6) {\mybox{?}};
\node at (1.5,0.6) {\mybox{?}};
\node at (2,0.6) {\mybox{?}};

\node at (0,1.2) {\mybox{?}};
\node at (0.5,1.2) {\mybox{?}};
\node at (1,1.2) {\mybox{?}};
\node at (1.5,1.2) {\mybox{?}};
\node at (2,1.2) {\mybox{?}};

\node at (0,1.8) {\mybox{?}};
\node at (0.5,1.8) {\mybox{?}};
\node at (1,1.8) {\mybox{?}};
\node at (1.5,1.8) {\mybox{?}};
\node at (2,1.8) {\mybox{?}};

\node at (0,2.4) {\mybox{?}};
\node at (0.5,2.4) {\mybox{?}};
\node at (1,2.4) {\mybox{?}};
\node at (1.5,2.4) {\mybox{?}};
\node at (2,2.4) {\mybox{?}};

\node at (-0.4,0.6) {4};
\node at (-0.4,1.2) {3};
\node at (-0.4,1.8) {2};
\node at (-0.4,2.4) {1};

\node at (0,2.9) {3};
\node at (0.5,2.9) {1};
\node at (1,2.9) {5};
\node at (1.5,2.9) {4};
\node at (2,2.9) {2};
\end{tikzpicture}
\caption{An example of a pile-scramble shuffle on a $4 \times 5$ matrix}
\label{fig2}
\end{figure}

\subsection{Chosen Pile Cut Protocol} \label{chosen}
Given a sequence of $q$ face-down stacks $A = (a_1,a_2,...,a_q)$, with each stack having the same number of cards, a \textit{chosen pile cut protocol} \cite{koch} enables $P$ to select a stack $a_i$ he/she wants without revealing $i$ to $V$. This protocol also reverts the sequence $A$ back to its original state after $P$ finishes using $a_i$ in other protocols.

\begin{figure}
\centering
\begin{tikzpicture}
\node at (0.0,2.4) {$\mystack{?}$};
\node at (0.8,2.4) {$\mystack{?}$};
\node at (1.6,2.4) {...};
\node at (2.4,2.4) {$\mystack{?}$};
\node at (3.2,2.4) {$\mystack{?}$};
\node at (4.0,2.4) {$\mystack{?}$};
\node at (4.8,2.4) {...};
\node at (5.6,2.4) {$\mystack{?}$};

\node at (0.0,2) {$a_1$};
\node at (0.8,2) {$a_2$};
\node at (2.4,2) {$a_{i-1}$};
\node at (3.2,2) {$a_i$};
\node at (4.0,2) {$a_{i+1}$};
\node at (5.6,2) {$a_q$};

\node at (0.0,1.4) {\mybox{?}};
\node at (0.8,1.4) {\mybox{?}};
\node at (1.6,1.4) {...};
\node at (2.4,1.4) {\mybox{?}};
\node at (3.2,1.4) {\mybox{?}};
\node at (4.0,1.4) {\mybox{?}};
\node at (4.8,1.4) {...};
\node at (5.6,1.4) {\mybox{?}};

\node at (0.0,1) {0};
\node at (0.8,1) {0};
\node at (2.4,1) {0};
\node at (3.2,1) {1};
\node at (4.0,1) {0};
\node at (5.6,1) {0};

\node at (0.0,0.4) {\mybox{?}};
\node at (0.8,0.4) {\mybox{?}};
\node at (1.6,0.4) {...};
\node at (2.4,0.4) {\mybox{?}};
\node at (3.2,0.4) {\mybox{?}};
\node at (4.0,0.4) {\mybox{?}};
\node at (4.8,0.4) {...};
\node at (5.6,0.4) {\mybox{?}};

\node at (0.0,0) {1};
\node at (0.8,0) {0};
\node at (2.4,0) {0};
\node at (3.2,0) {0};
\node at (4.0,0) {0};
\node at (5.6,0) {0};
\end{tikzpicture}
\caption{A $3 \times q$ matrix $M$ constructed in Step 1 of the chosen pile cut protocol}
\label{fig3}
\end{figure}

\begin{enumerate}
	\item Construct the following $3 \times q$ matrix $M$ (see Fig. \ref{fig3}).
	\begin{enumerate}
		\item In Row 1, publicly place the sequence $A$.
		\item In Row 2, secretly place a face-down sequence $E_q(i)$.
		\item In Row 3, publicly place a face-down sequence $E_q(1)$.
	\end{enumerate}
	\item Apply the pile-shifting shuffle to $M$.
	\item Turn over all cards in Row 2. Locate the position of the only \mybox{1}. A stack in Row 1 directly above this \mybox{1} will be the stack $a_i$ as desired.
	\item After finishing using $a_i$ in other protocols, place $a_i$ back into $M$ at the same position.
	\item Turn over all face-up cards. Apply the pile-shifting shuffle to $M$.
	\item Turn over all cards in Row 3. Locate the position of the only \mybox{1}. Shift the columns of $M$ cyclically such that this \mybox{1} moves to Column 1. This reverts $M$ back to its original state.
\end{enumerate}

The next two subprotocols are our newly developed subprotocols that will be used in our main protocol.

\subsection{Chosen $k$-Pile Cut Protocol} \label{chosen2}
A \textit{chosen $k$-pile cut protocol} is a generalization of the chosen pile cut protocol. Instead of selecting one stack, $P$ wants to select $k \leq q$ stacks $a_{\gamma_1},a_{\gamma_2},...,a_{\gamma_k}$ (the order of selection matters). This protocol enables $P$ to do so without revealing $\gamma_1,\gamma_2,...,\gamma_k$ to $V$.

\begin{figure}
\centering
\begin{tikzpicture}
\node at (0.0,2.4) {$\mystack{?}$};
\node at (0.8,2.4) {$\mystack{?}$};
\node at (1.6,2.4) {...};
\node at (2.4,2.4) {$\mystack{?}$};
\node at (3.2,2.4) {$\mystack{?}$};
\node at (4.0,2.4) {$\mystack{?}$};
\node at (4.8,2.4) {...};
\node at (5.6,2.4) {$\mystack{?}$};
\node at (6.4,2.4) {$\mystack{?}$};
\node at (7.2,2.4) {$\mystack{?}$};
\node at (8.0,2.4) {...};
\node at (8.8,2.4) {$\mystack{?}$};

\node at (0.0,2) {$a_1$};
\node at (0.8,2) {$a_2$};
\node at (2.4,2) {$a_{\gamma_1-1}$};
\node at (3.2,2) {$a_{\gamma_1}$};
\node at (4.0,2) {$a_{\gamma_1+1}$};
\node at (5.6,2) {$a_{\gamma_2-1}$};
\node at (6.4,2) {$a_{\gamma_2}$};
\node at (7.2,2) {$a_{\gamma_2+1}$};
\node at (8.8,2) {$a_q$};

\node at (0.0,1.4) {\mybox{?}};
\node at (0.8,1.4) {\mybox{?}};
\node at (1.6,1.4) {...};
\node at (2.4,1.4) {\mybox{?}};
\node at (3.2,1.4) {\mybox{?}};
\node at (4.0,1.4) {\mybox{?}};
\node at (4.8,1.4) {...};
\node at (5.6,1.4) {\mybox{?}};
\node at (6.4,1.4) {\mybox{?}};
\node at (7.2,1.4) {\mybox{?}};
\node at (8.0,1.4) {...};
\node at (8.8,1.4) {\mybox{?}};

\node at (0.0,1) {$0$};
\node at (0.8,1) {$0$};
\node at (2.4,1) {$0$};
\node at (3.2,1) {$1$};
\node at (4.0,1) {$0$};
\node at (5.6,1) {$0$};
\node at (6.4,1) {$2$};
\node at (7.2,1) {$0$};
\node at (8.8,1) {$0$};

\node at (0.0,0.4) {\mybox{?}};
\node at (0.8,0.4) {\mybox{?}};
\node at (1.6,0.4) {...};
\node at (2.4,0.4) {\mybox{?}};
\node at (3.2,0.4) {\mybox{?}};
\node at (4.0,0.4) {\mybox{?}};
\node at (4.8,0.4) {...};
\node at (5.6,0.4) {\mybox{?}};
\node at (6.4,0.4) {\mybox{?}};
\node at (7.2,0.4) {\mybox{?}};
\node at (8.0,0.4) {...};
\node at (8.8,0.4) {\mybox{?}};

\node at (0.0,0) {$1$};
\node at (0.8,0) {$2$};
\node at (2.4,0) {$\gamma_1-1$};
\node at (3.2,0) {$\gamma_1$};
\node at (4.0,0) {$\gamma_1+1$};
\node at (5.6,0) {$\gamma_2-1$};
\node at (6.4,0) {$\gamma_2$};
\node at (7.2,0) {$\gamma_2+1$};
\node at (8.8,0) {$q$};
\end{tikzpicture}
\caption{A $3 \times q$ matrix $M$ constructed in Step 1 of the $k$-chosen pile cut protocol for the case $k=2$}
\label{fig4}
\end{figure}

\begin{enumerate}
	\item Construct the following $3 \times q$ matrix $M$ (see Fig. \ref{fig4}).
	\begin{enumerate}
		\item In Row 1, publicly place the sequence $A$.
		\item In Row 2, secretly place a face-down $\mybox{i}$ at Column $\gamma_i$ for each $i=1,2,...,k$. Then, secretly place face-down $\mybox{0}$s at the rest of the columns.
		\item In Row 3, publicly place a face-down $\mybox{i}$ at Column $i$ for each $i=1,2,...,q$.
	\end{enumerate}
	\item Apply the pile-scramble shuffle to $M$.
	\item Turn over all cards in Row 2. Locate the position of an \mybox{$i$} for each $i=1,2,...,k$. A stack in Row 1 directly above the \mybox{$i$} will be the stack $a_{\gamma_i}$ as desired.
	\item After finishing using the selected stacks in other operations, place them back into $M$ at the same positions.
	\item Turn over all face-up cards. Apply the pile-scramble shuffle to $M$.
	\item Turn over all cards in Row 3. Rearrange the columns of $M$ such that the cards in Row 3 are \mybox{1}, \mybox{2}, ..., \mybox{$q$} in this order from left to right. This reverts $M$ back to its original state.
\end{enumerate}

\subsection{Color Checking Protocol} \label{color}
Given two sequences $E_q(x_1)$ and $E_q(x_2)$ ($0 \leq x_1,x_2 \leq q+1$), a \textit{color checking protocol} verifies that
\begin{enumerate}
	\item $1 \leq x_1 \leq q$, and
	\item either $x_2=x_1$ or $x_2=q+1$
\end{enumerate}
without revealing any other information.

\begin{enumerate}
	\item Construct the following $3 \times q$ matrix $M$.
	\begin{enumerate}
		\item In Row 1, publicly place the sequence $E_q(x_1)$.
		\item In Row 2, publicly place the sequence $E_q(x_2)$.
		\item In Row 3, publicly place a face-down sequence $E_q(1)$.
	\end{enumerate}
	\item Apply the pile-shifting shuffle to $M$.
	\item Turn over all cards in Row 1 to show that there are exactly $q-1$ $\mybox{0}$s and one $\mybox{1}$ (otherwise $V$ rejects). Suppose the only \mybox{1} is located at Column $j$.
	\item Turn over the card at Row 2 Column $j$ to show that it is a \mybox{1} (otherwise $V$ rejects).
	\item Turn over all face-up cards. Apply the pile-shifting shuffle to $M$.
	\item Turn over all cards in Row 3. Locate the position of the only \mybox{1}. Shift the columns of $M$ cyclically such that this \mybox{1} moves to Column 1. This reverts $M$ back to its original state.
\end{enumerate}

\section{Main Protocol}
We denote each color of the balls by number $1,2,...,n$. For each empty space in a bin, we put a ``dummy ball'' with number 0 in it, so every bin becomes full. Moreover, we put a dummy ball with number $n+1$ under each bin, and a dummy ball with number 0 right above each bin. Finally, we create an $(h+2) \times (n+m)$ matrix $M$ of stacks of cards, with a stack $E_n(i)$ representing a ball with number $i$. See Fig. \ref{fig5} for an example.

\begin{figure}
\centering
\begin{tikzpicture}
\draw[line width=0.6mm] (0,0) -- (0,1.9);
\draw[line width=0.6mm] (0,0) -- (0.7,0);
\draw[line width=0.6mm] (0.7,0) -- (0.7,1.9);
\draw[line width=0.6mm] (1,0) -- (1,1.9);
\draw[line width=0.6mm] (1,0) -- (1.7,0);
\draw[line width=0.6mm] (1.7,0) -- (1.7,1.9);
\draw[line width=0.6mm] (2,0) -- (2,1.9);
\draw[line width=0.6mm] (2,0) -- (2.7,0);
\draw[line width=0.6mm] (2.7,0) -- (2.7,1.9);
\draw[line width=0.6mm] (3,0) -- (3,1.9);
\draw[line width=0.6mm] (3,0) -- (3.7,0);
\draw[line width=0.6mm] (3.7,0) -- (3.7,1.9);

\node[draw,circle] at (0.35,0.35) {3};
\node[draw,circle] at (0.35,0.95) {1};
\node[draw,circle] at (0.35,1.55) {2};
\node[draw,circle] at (1.35,0.35) {2};
\node[draw,circle] at (1.35,0.95) {2};
\node[draw,circle] at (1.35,1.55) {3};
\node[draw,circle] at (2.35,0.35) {1};
\node[draw,circle] at (2.35,0.95) {1};
\node[draw,circle] at (3.35,0.35) {3};
\node at (0.0,-0.8) {};
\node at (4.3,0.95) {\LARGE{$\rightarrow$}};
\end{tikzpicture}
\begin{tikzpicture}
\draw[line width=0.6mm] (0,0) -- (0,1.9);
\draw[line width=0.6mm] (0,0) -- (0.7,0);
\draw[line width=0.6mm] (0.7,0) -- (0.7,1.9);
\draw[line width=0.6mm] (1,0) -- (1,1.9);
\draw[line width=0.6mm] (1,0) -- (1.7,0);
\draw[line width=0.6mm] (1.7,0) -- (1.7,1.9);
\draw[line width=0.6mm] (2,0) -- (2,1.9);
\draw[line width=0.6mm] (2,0) -- (2.7,0);
\draw[line width=0.6mm] (2.7,0) -- (2.7,1.9);
\draw[line width=0.6mm] (3,0) -- (3,1.9);
\draw[line width=0.6mm] (3,0) -- (3.7,0);
\draw[line width=0.6mm] (3.7,0) -- (3.7,1.9);

\node[draw,circle] at (0.35,-0.35) {4};
\node[draw,circle] at (0.35,0.35) {3};
\node[draw,circle] at (0.35,0.95) {1};
\node[draw,circle] at (0.35,1.55) {2};
\node[draw,circle] at (0.35,2.15) {0};
\node[draw,circle] at (1.35,-0.35) {4};
\node[draw,circle] at (1.35,0.35) {2};
\node[draw,circle] at (1.35,0.95) {2};
\node[draw,circle] at (1.35,1.55) {3};
\node[draw,circle] at (1.35,2.15) {0};
\node[draw,circle] at (2.35,-0.35) {4};
\node[draw,circle] at (2.35,0.35) {1};
\node[draw,circle] at (2.35,0.95) {1};
\node[draw,circle] at (2.35,1.55) {0};
\node[draw,circle] at (2.35,2.15) {0};
\node[draw,circle] at (3.35,-0.35) {4};
\node[draw,circle] at (3.35,0.35) {3};
\node[draw,circle] at (3.35,0.95) {0};
\node[draw,circle] at (3.35,1.55) {0};
\node[draw,circle] at (3.35,2.15) {0};
\node at (0.0,-0.8) {};
\end{tikzpicture}\\
\begin{tikzpicture}
\node at (0,4.4) {$\mystack{?}$};
\node at (1,4.4) {$\mystack{?}$};
\node at (2,4.4) {$\mystack{?}$};
\node at (3,4.4) {$\mystack{?}$};

\node at (0,4) {$E_3(0)$};
\node at (1,4) {$E_3(0)$};
\node at (2,4) {$E_3(0)$};
\node at (3,4) {$E_3(0)$};

\node at (0,3.4) {$\mystack{?}$};
\node at (1,3.4) {$\mystack{?}$};
\node at (2,3.4) {$\mystack{?}$};
\node at (3,3.4) {$\mystack{?}$};

\node at (0,3) {$E_3(2)$};
\node at (1,3) {$E_3(3)$};
\node at (2,3) {$E_3(0)$};
\node at (3,3) {$E_3(0)$};

\node at (0,2.4) {$\mystack{?}$};
\node at (1,2.4) {$\mystack{?}$};
\node at (2,2.4) {$\mystack{?}$};
\node at (3,2.4) {$\mystack{?}$};

\node at (0,2) {$E_3(1)$};
\node at (1,2) {$E_3(2)$};
\node at (2,2) {$E_3(1)$};
\node at (3,2) {$E_3(0)$};

\node at (0,1.4) {$\mystack{?}$};
\node at (1,1.4) {$\mystack{?}$};
\node at (2,1.4) {$\mystack{?}$};
\node at (3,1.4) {$\mystack{?}$};

\node at (0,1) {$E_3(3)$};
\node at (1,1) {$E_3(2)$};
\node at (2,1) {$E_3(1)$};
\node at (3,1) {$E_3(3)$};

\node at (0,0.4) {$\mystack{?}$};
\node at (1,0.4) {$\mystack{?}$};
\node at (2,0.4) {$\mystack{?}$};
\node at (3,0.4) {$\mystack{?}$};

\node at (0,0) {$E_3(4)$};
\node at (1,0) {$E_3(4)$};
\node at (2,0) {$E_3(4)$};
\node at (3,0) {$E_3(4)$};
\end{tikzpicture}
\caption{The way we put dummy balls into an instance of a puzzle, and create a matrix of stacks of cards}
\label{fig5}
\end{figure}

The key observation is that, when moving a ball represented by $a_x$ from bin (column) $A$ to bin (column) $B$, we can view it as swapping $a_x$ with a dummy ball represented by $b_y$ occupying the corresponding empty space in bin $B$. In each such move, $P$ performs the following steps.

\begin{enumerate}
	\item Put all stacks in each column of $M$ together into a single big stack. (We now have $n+m$ big stacks of cards.) Apply the $k$-chosen pile cut protocol with $k=2$ to select columns $A$ and $B$ from $M$.
	\item Arrange the cards in column $A$ as a sequence of stacks $(a_1,a_2,...,a_{h+2})$ representing the $h+2$ balls, where $a_1=E_n(0)$ is the topmost stack in $A$ and $a_{h+2}=E_n(n+1)$ is the bottommost stack in $A$. Analogously, arrange the cards in $B$ as $(b_1,b_2,...,b_{h+2})$.
	\item Apply the chosen pile cut protocol to select a stack $a_x$ from column $A$ and a stack $b_y$ from column $B$. Note that the chosen pile cut protocol preserves the cyclic order of the sequence, thus we can refer to stacks $a_{x-1}$, $b_{y-1}$, and $b_{y+1}$ (where the indices are considered modulo $n$).
	\item Turn over $a_{x-1}$, $b_{y-1}$, and $b_y$ to reveal that they are all $E_n(0)$s (otherwise $V$ rejects).
	\item Apply the color checking protocol on $a_x$ and $b_{y+1}$.
	\item Swap $a_x$ and $b_y$.
\end{enumerate}

$P$ performs the above steps for each of the $t$ moves in $P$'s solution. After finishing all moves, $P$ applies the pile-scramble shuffle to the columns of $M$ to shuffle all columns into a random permutation. Finally, $P$ turns over all cards in $M$. $V$ verifies that there is one column consisting of $(E_n(0),E_n(i),E_n(i),...,E_n(i),$ $E_n(n+1))$ for each $i=1,2,...,n$, and there are $m$ columns each consisting of $(E_n(0),E_n(0),...,E_n(0),E_n(n+1))$. If the verification passes, $V$ accepts; otherwise, $V$ rejects.

Our protocol uses a total of $\Theta(hn(n+m))$ cards.

\section{Proof of Correctness and Security}
We will prove the perfect completeness, perfect soundness, and zero-knowledge properties of our protocol.

\begin{lemma}[Perfect Completeness] \label{lem1}
If $P$ knows a solution with $t$ moves of the ball sort puzzle, then $V$ always accepts.
\end{lemma}

\begin{proof}
First, we will prove the correctness of our newly developed subprotocols: the chosen $k$-pile cut protocol in Section \ref{chosen2} and the color checking protocol in Section \ref{color}.

Consider the chosen $k$-pile cut protocol. In Step 1(b), for each $i=1,2,...,k$, $P$ places an $\mybox{i}$ in the same column as the stack $a_{\gamma_i}$. After applying the pile-scramble shuffle, they will still be in the same column. Therefore, the stacks selected in Step 3 are $a_{\gamma_1},a_{\gamma_2},...,a_{\gamma_k}$ in this order as desired.

Consider the color checking protocol. In Step 3, $V$ verifies that $E(x_1)$ consists of exactly one $\mybox{1}$ and $q-1$ $\mybox{0}$s, which is an equivalent condition to $1 \leq x_1 \leq q$. In Step 4, $V$ verifies that the $x_1$-th leftmost card of $E(x_2)$ is a \mybox{1}, which is an equivalent condition to $x_2=x_1$ or $x_2=q+1$. Therefore, the protocol is correct.

Now suppose $P$ knows a solution with $t$ moves of the puzzle.

Consider each move of a ball represented by $a_x$ in bin (column) $A$ to an empty position represented by $b_y$ in bin (column) $B$. Since this is a valid move, the moved ball must be a topmost real ball in bin $A$, so $a_{x-1}$ must be an $E_n(0)$ (representing either a dummy ball inside the bin or a dummy ball above the bin). As $b_y$ represents a dummy ball inside bin $B$ (not above the bin), $b_y$ and $b_{y-1}$ must also both be $E_n(0)$s. Hence, Step 4 of the main protocol will pass.

Also, since this is a valid move, $a_x$ must represent a real ball, i.e. be an $E_n(\alpha)$ for some $1 \leq \alpha \leq n$, and bin $B$ must either be empty or has the topmost real ball with the same color as the moved ball. In the former case, we have $b_{y+1}=E_n(n+1)$; in the latter case, we have $b_{y+1}=E_n(\alpha)$. Hence, the color cheking protocol will pass.

After $t$ moves, the balls are sorted by color such that each bin is either empty or full with balls of a single color. Hence, the final step of the main protocol will pass.

Therefore, $V$ always accepts.
\end{proof}

\begin{lemma}[Perfect Soundness] \label{lem2}
If $P$ does not know a solution with $t$ moves of the ball sort puzzle, then $V$ always rejects.
\end{lemma}

\begin{proof}
We will prove the contrapositive of this statement. Suppose that $V$ accepts, meaning that the verification passes for every move, and the final step of the main protocol also passes.

Consider each move of swapping stacks $a_x$ in column $A$ and $b_y$ in column $B$. Since the color checking protocol on $a_x$ and $b_{y+1}$ passes, we have that $a_x$ represents a ball with number $\alpha$ for some $1 \leq \alpha \leq n$, i.e. a real ball (not a dummy ball). Since Step 4 of the main protocol passes, we have that $a_{x-1}$ is $E_n(0)$. That means $a_x$ represents a topmost real ball in bin $A$.

Since Step 4 of the main protocol passes, we also have that $b_y$ and $b_{y-1}$ are both $E_n(0)$s, which means $b_y$ represents a dummy ball inside bin $B$ (not above the bin), so bin $B$ is not full. Since the color checking protocol on $a_x$ and $b_{y+1}$ passes, we have that $b_{y+1}$ represents either a real ball with the same number $\alpha$ or a dummy ball with number $n+1$. In the former case, the topmost real ball in bin $B$ has the same color as the moved ball; in the latter case, bin $B$ is empty. Hence, we can conclude that this move is a valid move in the puzzle.

Since the final step of the main protocol passes, we have that after $t$ valid moves, the balls are sort by color such that each bin is either empty or full with balls of a single color. Therefore, we can conclude that $P$ knows a solution with $t$ moves of the puzzle.
\end{proof}

\begin{lemma}[Zero-Knowledge] \label{lem3}
During the verification, $V$ learns nothing about $P$'s solution.
\end{lemma}

\begin{proof}
It is sufficient to show that all distributions of cards that are turned face-up can be simulated by a simulator $S$ that does not know $P$'s solution.

\begin{itemize}
	\item In Steps 3 and 6 of the chosen pile cut protocol in Section \ref{chosen}, due to the pile-shifting shuffle, the \mybox{1} has an equal probability to be at any of the $q$ positions. Hence, these steps can be simulated by $S$.
	\item In Steps 3 and 6 of the chosen $k$-pile cut protocol in Section \ref{chosen2}, due to the pile-scramble shuffle, the order of the $q$ cards are uniformly distributed among all $q!$ permutations of them. Hence, these steps can be simulated by $S$.
	\item In Steps 3 and 6 of the color checking protocol in Section \ref{color}, due to the pile-shifting shuffle, the \mybox{1} has an equal probability to be at any of the $q$ positions. Hence, these steps can be simulated by $S$.
	\item In the final step of the main protocol, there is one column of $M$ consisting of $(E_n(0),$ $E_n(i),E_n(i),...,E_n(i),E_n(n+1))$ for each $i=1,2,...,n$, and there are $m$ columns each consisting of $(E_n(0),E_n(0),...,E_n(0),E_n(n+1))$. Due to the pile-scramble shuffle, the order of the $n+m$ columns are uniformly distributed among all $(n+m)!$ permutations of them. Hence, these steps can be simulated by $S$.
\end{itemize}

Therefore, we can conclude that $V$ learns nothing about $P$'s solution.
\end{proof}

\section{Future Work}
We developed a card-based ZKP protocol for the ball sort puzzle. A possible future work is to develop such protocol for water sort puzzle, a similar puzzle with more restrictive rules. (In the water sort puzzle, consecutive ``balls'' with the same color are connected and must be moved together.) We believe it is significantly harder to construct a ZKP for this puzzle as $P$ must also hide the number of moved balls in each move.

\subsubsection*{Acknowledgement}
The author would like to thank Daiki Miyahara for a valuable discussion on this research.

\end{document}